\documentclass{article}
\usepackage{amssymb}
\usepackage{amsmath}
\usepackage{amsthm}
\usepackage{amsfonts}
\usepackage{graphicx}

\newtheorem{theorem}{Theorem}

\newtheorem{example}[theorem]{Example}

\newtheorem{lemma}[theorem]{Lemma}

\renewenvironment{proof}[1][Proof]{\noindent\textbf{#1.} }{\
\rule{0.5em}{0.5em}}
\numberwithin{equation}{section}
\input{tcilatex}
\begin{document}

\title{Don't stay local - extrapolation analytics for Dupire's local
volatility}
\author{Peter Friz\footnote{
  partially supported by MATHEON
}
, Stefan Gerhold \\
%EndAName
TU\ and WIAS Berlin, TU Wien}
\maketitle

\begin{abstract}
A robust implementation of a Dupire type local volatility model is an
important issue for every option trading floor. Typically, this (inverse)
problem is solved in a two step procedure : (i)\ a smooth parametrization of
the implied volatility surface; (ii) computation of the local volatility
based on the resulting call price surface. Point (i), and in particular how
to extrapolate the implied volatility in extreme strike regimes not seen in
the market, has been the subject of numerous articles, starting with Lee
(Math.\ Finance, 2004). In the present paper we give direct analytic
insights into the asymptotic behavior of local volatility at extreme strikes.
\end{abstract}

\section{A new formula for local volatility extrapolation}
\label{se:intro}

Volatility remains a key concept in modern quantitative finance. In
particular, the Black-Scholes \textit{implied volatility} surface $\sigma_{%
\mathrm{BS}}=\sigma_{\mathrm{BS}}(K,T)$ is the central object of any option
trading desk, see e.g.~\cite{Ga06}. On the quantitative and computational
side, a smooth and arbitrage free parametrization of the implied volatility
surface is a crucial step towards a robust implementation of a Dupire type
local volatility model~\cite{Du94, Du96}. Indeed, Dupire's formula%
\begin{equation}  \label{eq:dupire}
\sigma _{\mathrm{loc}}^{2}( K,T) =\frac{2\partial _{T}C}{K^{2}\partial _{KK}C%
}
\end{equation}
implies that any arbitrage free call price surface 
\begin{equation*}
C=C( K,T) =  C_{\mathrm{BS}}( K,T;\sigma _{\mathrm{BS}}( K,T) ) 
\end{equation*}
which arises from a (not necessarily Markovian) It\^{o} diffusion is
obtained from the one-factor (``Dupire's local vol'') model%
\begin{equation*}
dS_{t}/S_{t}=\sigma _{\mathrm{loc}}( S_{t},t) dW_{t}.
\end{equation*}%
(Note that spot remains fixed in the present discussion, and that we work
under the appropriate forward measure to avoid drift terms.) It is helpful
to think of local volatility as a Markovian projection, term coined in \cite%
{Pi06}, of a higher dimensional model (e.g.\ Heston); the first component
then forms an It\^{o} diffusion of the form 
\begin{equation*}
dS_{t}/S_{t}=\sigma _{\mathrm{stoch }}( t,\omega ) dW_{t}.
\end{equation*}%
Indeed, it is known (see e.g.~\cite{Ga06} and the references therein) that $%
\sigma _{\mathrm{loc}}^{2}( K,T) =\mathbb{E}[ \sigma _{\text{stoch }%
}^{2}|S_{T}=K] $; in practice, this means that even for stochastic
volatility models with fully explicit Markovian specification, sampling from
the corresponding local volatility models requires substantial computational
effort. (In particular, the singular conditioning requires Malliavin
calculus techniques, as was pointed out e.g.\ in~\cite{HL09}.)

The analysis of implied, local, and stochastic volatility and their
interplay has been subject of countless works; a very small selection
relevant to the present discussion is \cite{AnPi07, AvBOBuFr03, BeBuFl04,
FrGeGuSt11, HL05, Le04a}. Our contribution here is a formula
(\eqref{LocVolLee} below)
that allows for approximation of $\sigma _{\mathrm{loc}}^{2}(K,T)$ when $K$
is large (and similarly, $K$ is small). The main ingredient to this formula
is a known moment generating function (mgf) of the log-price $(X_{t})$
(under the pricing measure), 
\begin{equation*}
M(s,T):=\exp (m(s,T)):=\mathbb{E}\exp (sX_{T}),
\end{equation*}%
assumed to be finite in some (maximal) interval $(s_{-}(T),s_{+}(T))$ with
critical exponents $s_{-}$ and $s_{+}$ defined as%
\begin{equation*}
s_{-}(T):=\inf \left\{ s:M(s,T)<\infty \right\} ,\qquad s_{+}(T):=\sup
\left\{ s:M(s,T)<\infty \right\} .
\end{equation*}%
We also assume that call prices have sufficient regularity to make Dupire's
formula~\eqref{eq:dupire} well-defined, and that the mgf blows up at the
upper critical moment:
\begin{equation}\label{eq:blowup}
\lim_{s\uparrow s_{+}(T)}M(s,T)=\infty .
\end{equation}
This holds, e.g., in the Heston model~\cite{He93}, with log-price $X_{t}=\log
(S_{t}/S_{0})$, where 
\begin{align*}
dS_{t}& =S_{t}\sqrt{Y_{t}}dW_{t},\qquad S_{0}=s_{0}>0, \\
dV_{t}& =(a+bV_{t})dt+c\sqrt{V_{t}}dZ_{t},\qquad V_{0}=v_{0}>0,
\end{align*}%
with $a\geq 0$, $b\leq 0$, $c>0$, and $d\left\langle W,Z\right\rangle
_{t}=\rho dt$ with $\rho \in (-1,1)$. We will prove the following theorem,
which is reminiscent of Lee's formula~\cite{Le04a} for implied volatility.
\begin{theorem}
\label{thm:mainHeston} In the Heston model with $\rho \leq 0$ (the relevant
regime in practice), the following local volatility approximation holds:%
\footnote{
  By a common abuse of notation, we write $\sigma _{\mathrm{loc}}^{2}(k,T)$
  instead of $\sigma _{\mathrm{loc}}^{2}(e^k,T)$ when we wish to
  express the local vol as a function of log-strike~$k$.
}
\begin{equation}\label{eq:heston asympt}
\lim_{k\rightarrow \infty }\frac{\sigma _{\mathrm{loc}}^{2}(k,T)}{k}=\frac{2%
}{T\,s_{+}(s_{+}-1)R_{1}/R_{2}},
\end{equation}%
where $k=\log \left( K/S_{0}\right) ,$ $s_{+}\equiv s_{+}(T)$ and 
\begin{align}
R_{1}& =c^{2}s_{+}(s_{+}-1)\left[ c^{2}(2s_{+}-1)-2\rho c(s_{+}\rho c+b)%
\right]  \label{eq:R1} \\
& \quad -2(s_{+}\rho c+b)\left[ c^{2}(2s_{+}-1)-2\rho c(s_{+}\rho c+b)\right]
\notag \\
& \quad +4\rho c\left[ c^{2}s_{+}(s_{+}-1)-(s_{+}\rho c+b)^{2}\right] , \notag \\
R_{2}& =2c^{2}s_{+}(s_{+}-1)\left[ c^{2}s_{+}(s_{+}-1)-(s_{+}\rho c+b)^{2}%
\right]. \label{eq:R2}
\end{align}
\end{theorem}
The origin of this result lies in the saddle point based approximation
formula 
\begin{equation}
\sigma _{\mathrm{loc}}^{2}(k,T)\approx \left. \frac{2\frac{\partial }{%
\partial T}m(s,T)}{s(s-1)}\right\vert _{s=\hat{s}(k,T)},  \label{LocVolLee}
\end{equation}%
where $k$ denotes log-strike, and $\hat{s}=\hat{s}(k,T)$ is determined as
solution of the equation%
\begin{equation}
\frac{\partial }{\partial s}m(s,T)=k.  \label{ApproximateSaddlePoint}
\end{equation}%
As such, our formula \eqref{LocVolLee} is not restricted to the Heston
model. As a trivial example, let us consider the generalized Black-Scholes
model with time-dependent volatility,%
\begin{equation*}
dS_{t}=S_{t}\sqrt{v(t)}dW_{t}.
\end{equation*}%
We find $m(s,T)=\frac{1}{2}s(s-1)\int_{0}^{T}v(t)dt$, and then, correctly,%
\begin{equation*}
\sigma _{\mathrm{loc}}^{2}(k,T)=v(T).
\end{equation*}%
(In this example the evaluation of $\hat{s}(k,T)$ plays no role, since the
fraction on the right hand side of~\eqref{LocVolLee} does not depend on it.) 

\bigskip 

In fact, we expect our approximation formula~\eqref{LocVolLee} to work whenever
the saddle
point method is applicable (also assuming that call prices are smooth enough
to make~\eqref{eq:dupire} a well-defined quantity); the essence of the
argument is given in Section~\ref{se:saddle}. Of course, the ultimate
justification of a saddle point approximation involves tail estimates which
may present mathematical challenges (while easy to observe
numerically); in the Heston case,  we
achieve this by a subtle application of ODE comparison results, applied to
the underlying Riccati equations (cf.\ Appendix~\ref{se:app heston}), thus
completing our proof of the above theorem.
The asymptotic equivalence of~\eqref{eq:heston asympt} and~\eqref{LocVolLee}
is discussed in Section~\ref{se:heston}.

Interestingly, even when the blow-up of the mgf is too slow to apply the saddle point
method, the approximation formula~\eqref{LocVolLee} can give surprisingly
accurate results. Our attempt to understand this phenomenon, besides
numerical evidence in the variance gamma model for suitable parameters,
discussed in Section~\eqref{se:jumps}, passes through Karamata's
Tauberian theorem and is
the content of Section~\ref{se:karamata}. We have not pushed our investigations too
far, however, since the meaning of Dupire's local volatility in the presence
of jumps may be questioned (cf.\ our comment below on extension of Dupire's
formula to jump settings.)

Various additional comments are in order.

\begin{enumerate}
\item Equation~\eqref{ApproximateSaddlePoint} is solvable for large~$k$,
since~\eqref{eq:blowup} implies 
\begin{equation*}
\lim_{s\uparrow s_{+}}\frac{\partial }{\partial s}m(s,T)=\infty .
\end{equation*}

\item If there is no blow-up, i.e.~\eqref{eq:blowup} does not hold, then~%
\eqref{LocVolLee} is typically incorrect. See Example~\ref{ex:NIG} in
Section~\ref{se:jumps} for some hints on how to handle such cases.

\item We have $\hat{s}(k,T)\uparrow s_{+}(T)$ as $k\rightarrow \infty $;
hence, in models with moment explosion \cite{AnPi07, Le04a}, where $%
s_{+}(T)<\infty $, the denominator in~\eqref{LocVolLee} may be replaced by $%
s_{+}(T)(s_{+}(T)-1)$. While this is correct to first order, it is often
preferable to use~\eqref{LocVolLee} as it is, and to calculate~$\hat{s}(k,T)$
by (numerically) solving~\eqref{ApproximateSaddlePoint}. This tradeoff
between simple formulas and numerical precision is illustrated in several
examples in Sections~\ref{se:heston} and~\ref{se:jumps}. The comment
applies in particular to the Heston model.

\item There is a version of our approximation formula~\eqref{LocVolLee} for \emph{small}
values of $K$ (i.e.\ $K\downarrow 0$, or $k\downarrow -\infty$), which
requires that
the mgf blows up at the lower critical moment~$s_{-}(T)$. If $k<0$ and $|k|$
is large, equation~\eqref{ApproximateSaddlePoint} has a unique solution $%
\hat{s}_{-}(k,T)<0$. Then the approximation~\eqref{LocVolLee} holds, if~$%
\hat{s}$ is replaced by $\hat{s}_{-}$. 
%\item In the examples below, local vol is $O(k)$ as $k\to\infty$.
% We expect that this upper bound holds for any L\'evy model
% that satisfies~\eqref{eq:blowup}. (But cf.\ the following comment on models
% with jumps.)

\item There are extensions of Dupire's work to jump diffusions and also pure
jump models; the resulting \textquotedblleft local\textquotedblright\
version of these models is studied in~\cite{BeCo10}. Local L\'{e}vy models
were introduced earlier in~\cite{CaGeMaYo04}. In particular, Dupire's
formula (which may be written as a PDE) becomes a PIDE which features an
integral term involving the second derivative of~$C$ w.r.t.\ strike, times a
kernel depending on $K$, integrated against all strikes in $(0,\infty )$.
(The formula, which we need not reproduce here in full technical detail,
appears in Theorem~1 of~\cite{BeCo10}.)

Another difficulty in the jump setting is the potential lack of immediate
smoothing. For instance, the variance gamma model satisfies the above PIDE
only in viscosity sense; in fact, call prices in the variance gamma model
may not be twice differentiable in $K$ for small times, as was noted in~\cite%
{CoVo05}. But for sufficiently large times our formula~\eqref{LocVolLee}
works, see Example~\ref{ex:var gamma} in Section~\ref{se:jumps}.

We conclude that, in a general jump setting, Dupire's formula, as stated in~%
\eqref{eq:dupire}, may be ill-defined; moreover, even if call prices are
smooth enough to make the formula well-defined, it fails to recreate the
correct marginals of the original price process.

\item Even so, given the industry practice of applying Dupire's formula to
any given call price surface, we discuss in Section~\ref{se:jumps} what
happens when applying~\eqref{LocVolLee} to jump models, if possible.
Formula~\eqref{LocVolLee} simplifies in exponential L\'{e}vy models, which
have the property that $m(s,T)$ is linear in $T$; thus, the numerator in~%
\eqref{LocVolLee} may be replaced by $2m(s,1)$. In jump models, we also
expect $\sigma _{\mathrm{loc}}^{2}(k,T)$ to explode as $T\downarrow 0$, and
we shall observe and quantify this blow-up in some examples below. There is
potential practical value in that a Dupire local volatility surface, fitted
to market data, may so be inspected for evidence of jump behavior (thereby
questioning the use of Dupire's formula in the first place).
\end{enumerate}

\section{Saddle point asymptotics}

\label{se:saddle}

As is well known~\cite{CaMa99}, we can recover the call price~$C$ and the
probability density~$D(\cdot,T)$ of~$S_T$ by Laplace-Fourier inversion from
the mgf: 
\begin{equation}  \label{eq:call}
C(K,T) = \frac{e^k}{2i\pi} \int_{-i\infty}^{i\infty} e^{-ks}\frac{M(s,T)}{%
s(s-1)} ds,
\end{equation}
\begin{equation}  \label{eq:density}
D(x,T) = \frac{1}{2i\pi} \int_{-i\infty}^{i\infty} e^{-(s+1)\log x}M(s,T) ds.
\end{equation}
Now differentiate the call price under the integral sign: 
\begin{equation}  \label{eq:call dot}
\partial_T C(K,T) = \frac{e^k}{2i\pi} \int_{-i\infty}^{i\infty} \frac{%
\partial_T m(s,T)}{s(s-1)} e^{-ks}M(s,T) ds.
\end{equation}
By Dupire's formula, we have 
\begin{equation}  \label{eq:dup frac}
\sigma_{\mathrm{loc}}^2(k,T) = \frac{2\partial_T C(K)}{ K^2 D(K,T)} = \frac{2
\int_{-i\infty}^{i\infty} \frac{\partial_T m(s,T)}{s(s-1)} e^{-ks}M(s,T) ds%
} {\int_{-i\infty}^{i\infty} e^{-ks} M(s,T) ds}.
\end{equation}
Both integrands in~\eqref{eq:dup frac} have a singularity at $s=s_+$, since $%
M(s,T)$ gets infinite there. The singular behavior of $M(s,T)$ dominates the
asymptotics of both integrals. The resulting asymptotic factor cancels, and
only the contribution of $2\frac{\partial_T m(s,T)}{s(s-1)}$ remains. This is
the idea behind~\eqref{LocVolLee}.

To implement it, we analyze both integrals in~\eqref{eq:dup frac} by a
saddle point approximation~\cite{deBr58}. If~$M$ features an exponential
blow-up at the critical moment~$s_+$, its validity can be justified rather
universally. Examples include the Heston model, double exponential L\'evy,
and Black-Scholes. (Note that the critical moment is $s_+=\infty$ for
Black-Scholes.) If the saddle point method is not applicable (because of
insufficient blow-up), different arguments are required; see the following
section.

So let us proceed with the saddle point analysis of~\eqref{eq:dup frac}. For
both integrals, we only use the factor $e^{-k s}M(s)$ to find the location
of the (approximate) saddle point~$\hat{s}=\hat{s}(k,T)$. The saddle point
equation is~\eqref{ApproximateSaddlePoint}, obtained by equating the
derivative of $e^{-k s}M(s)$ to zero. We move the integration contour
through the saddle point. Then, for large~$k$, only a small part $%
|\Im(s)|\leq h(k)$ of the contour, around the saddle point, matters
asymptotically. (The choice of the function~$h$ depends on the singular
expansion of~$M$; see Section~\ref{se:heston} for an example.) The integral
can be approximated via a local expansion of the integrand. Let us carry
this out for the denominator of~\eqref{eq:dup frac}. (In the following
formulas we write~$m^{\prime \prime }$ for $\partial^2 m/\partial s^2$.) 
\begin{align}
\int_{\hat{s}-i\infty}^{\hat{s}+i\infty} & e^{-ks} M(s,T) ds \sim \int_{\hat{%
s}-i h(k)}^{\hat{s}+i h(k)} e^{-ks} M(s,T) ds  \notag \\
&\sim \int_{\hat{s}-i h(k)}^{\hat{s}+i h(k)} \exp\left( -ks + m(\hat{s},T) +
k (s-\hat{s}) + \tfrac12 m^{\prime \prime }(\hat{s},T) (s-\hat{s})^2 \right)
ds  \notag \\
&= e^{m(\hat{s},T)-k\hat{s}} \int_{\hat{s}-i h(k)}^{\hat{s}+i h(k)}
\exp\left( \tfrac12 m^{\prime \prime }(\hat{s},T) (s-\hat{s})^2 \right) ds.
\label{eq:denominator}
\end{align}
In the Taylor expansion of the exponent we have used the equation $m^{\prime
}(\hat{s},T)=k$. Now the crucial observation is that the numerator of~%
\eqref{eq:dup frac} admits a similar approximation, where the only new
ingredient is the factor $2\frac{\partial_T m(s,T)}{s(s-1)}$: 
\begin{align}
2 & \int_{-i\infty}^{i\infty} \frac{\partial_T m(s,T)}{s(s-1)} e^{-ks}M(s,T)
ds  \notag \\
&\sim 2 \int_{\hat{s}-i h(k)}^{\hat{s}+i h(k)} \frac{\partial_T m(s,T)}{%
s(s-1)} e^{-ks}M(s,T) ds  \notag \\
&\sim 2 e^{m(\hat{s},T)-k\hat{s}} \int_{\hat{s}-i h(k)}^{\hat{s}+i h(k)} 
\frac{\partial_T m(\hat{s},T)}{\hat{s}(\hat{s}-1)} \left(1 + o(1) \right)
\exp\left( \tfrac12 m^{\prime \prime }(\hat{s},T) (s-\hat{s})^2 \right) ds 
\notag \\
&\sim 2 \frac{\partial_T m(\hat{s},T)}{\hat{s}(\hat{s}-1)} e^{m(\hat{s},T)-k%
\hat{s}} \int_{\hat{s}-i h(k)}^{\hat{s}+i h(k)} \exp\left( \tfrac12
m^{\prime \prime }(\hat{s},T) (s-\hat{s})^2 \right) ds.  \label{eq:numerator}
\end{align}
Dividing~\eqref{eq:numerator} by~\eqref{eq:denominator} concludes the
derivation. Summarizing, we note that the asymptotics of~$\sigma_{\mathrm{loc%
}}^2(k)$ are governed by the local expansions at $s=\hat{s}$ of the
integrands in~\eqref{eq:dup frac}. The respective first terms of both
expansions agree, and thus cancel, except for the factor~\eqref{LocVolLee}.

\section{Algebraic singularities and Karamata's theorem}

\label{se:karamata}

The saddle point method is well suited to treat mgfs of exponential growth,
such as $M(s,T) \approx \exp(1/(s_+-s))$, but fails in cases of slower
blow-up. To see how to analyze these, let us assume that the mgf~$M$ grows
like a power at the (finite) critical moment: 
\begin{equation*}
M(s,T) \sim \frac{c_1 }{(s_+-s)^{c_2}}, \qquad s\uparrow s_+. 
\end{equation*}
(The variance gamma model is a typical instance.) The quantities $%
c_1=c_1(T)>0$ and $c_2=c_2(T)>0$ are independent of~$s$, but may be
functions of maturity~$T$. (In particular, we assume that~$c_2$ does depend
on~$T$, which holds in L\'evy models.) Since 
\begin{equation*}
\frac{\partial}{\partial s}m(s,T) \sim \frac{c_2}{s_+-s}, 
\end{equation*}
the saddle point~$\hat{s}$ satisfies 
\begin{equation*}
\hat{s} \approx s_+ - \frac{c_2}{k}. 
\end{equation*}
Inserting this into the time derivative 
\begin{equation*}
\frac{\partial}{\partial T}m(s,T) \sim \dot{c}_2(T) \log \frac{1}{s_+-s}
\end{equation*}
of~$m$ yields 
\begin{equation}  \label{eq:goal}
\left. \frac{2\frac{\partial }{ \partial T}m\left( s,T\right) }{s\left(
s-1\right) }\right\vert _{s=\hat{s} \left( k,T\right) } \sim \frac{2 \dot{c}%
_2(T)\log k}{s_+(s_+-1)}.
\end{equation}
To justify~\eqref{LocVolLee}, we now have to argue that~$\sigma_{\mathrm{loc}%
}^2(k,T)$ has the same asymptotics. Again, we use the representation~%
\eqref{eq:dup frac}. To put it briefly, the reason why the approach from the
preceding section fails is that one cannot find a suitable~$h(k)$. (Either
the tails $|\Im(s)|>h(k)$ of the integrals are not negligible, or the local
expansion is not uniformly valid.) But it is still true that the local
behavior of the integrands near~$s_+$ fully determines the asymptotics of
the integrals in~\eqref{eq:dup frac}.

We write $f(\cdot,T)$ for the probability density of the log-price~$X_T$.
Note that the denominator in~\eqref{eq:dup frac} equals $2i\pi f(k,T)$. The
(one-sided!) Laplace transform of~$k \mapsto e^{s_+ k}f(k,T)$ satisfies 
\begin{equation}  \label{eq:lap}
\int_0^\infty e^{-sk} e^{s_+ k}f(k,T) dk \sim c_1 s^{c_2}, \qquad
s\downarrow 0.
\end{equation}
This follows from 
\begin{equation}  \label{eq:lap M}
M(s,T) = \int_{-\infty}^0 e^{sk} f(k,T) dk + \int_0^{\infty} e^{sk} f(k,T)
dk \sim \frac{c_1}{(s_+ - s)^{c_2}}, \qquad s\uparrow s_+,
\end{equation}
since the first integral in~\eqref{eq:lap M} is~$O(1)$. By Karamata's
Tauberian theorem~\cite[Theorem~1.7.1]{BiGoTe87}, we obtain from~%
\eqref{eq:lap} that 
\begin{equation*}
\int_0^k e^{s_+ x} f(x,T) dx \sim \frac{c_1 k^{c_2} }{\Gamma(c_2+1)}, \qquad
k\to\infty, 
\end{equation*}
hence, by differentiating, 
\begin{equation}  \label{eq:denom as}
f(k,T) \approx e^{-s_+ k} \frac{c_1 k^{c_2-1}}{\Gamma(c_2)}, \qquad
k\to\infty.
\end{equation}
Similarly, the asymptotics 
\begin{equation*}
\frac{\partial_T m(s,T)}{s(s-1)} M(s,T) \sim \frac{\dot{c}_2}{s_+(s_+ -1)} 
\log \frac{1}{s_+-s} \times  \frac{c_1}{(s_+ -s)^{c_2}}, \qquad s\uparrow
s_+, 
\end{equation*}
imply that the numerator in~\eqref{eq:dup frac} approximately equals 
\begin{equation}  \label{eq:num as}
\approx \frac{2 c_1 \dot{c}_2}{s_+(s_+ -1)\Gamma(c_2)} e^{-s_+ k} k^{c_2-1}
\log k.
\end{equation}
Now divide~\eqref{eq:num as} by~\eqref{eq:denom as} to see that $\sigma_{%
\mathrm{loc}}^2(k,T)$ approximately equals~\eqref{eq:goal}.
Note that we did not talk about Tauberian conditions, which are necessary to make this
derivation rigorous, such as monotonicity of the density. In concrete
cases, where an analytic continuation of the mgf is available, a Hankel
contour approach~\cite{Ford60} might be preferable to Karamata's theorem.
%An issue that we do not understand is why the saddle point based
%formula~\eqref{LocVolLee} seems to be more accurate than
%the right hand side of~\eqref{eq:goal} (suggested by numerical evidence).

\section{Local vol at extreme strikes in the Heston model}

\label{se:heston}

In this section we give a numerical example and explain how~\eqref{eq:heston asympt}
is obtained by specializing~\eqref{LocVolLee}.
(But recall that~\eqref{LocVolLee} is so far just a recipe and not a theorem;
a rigorous proof of~\eqref{eq:heston asympt} is given in Appendix~\ref{se:app heston}.)

Figure~\ref{fig:heston} compares the approximations~\eqref{eq:heston asympt}
and~\eqref{LocVolLee} for the local vol.
While asymptotically equivalent, the plot suggests that~\eqref{eq:heston asympt}
has an $O(1)$ error term, whereas the error of~\eqref{LocVolLee} seems to be
only~$o(1)$. Note that the right hand side of~\eqref{LocVolLee} can be easily evaluated
numerically, by using the explicit expression~\cite{He93} of the Heston mgf
in~\eqref{ApproximateSaddlePoint}.
\begin{figure}[ht]
\centering
\includegraphics[width=8truecm]{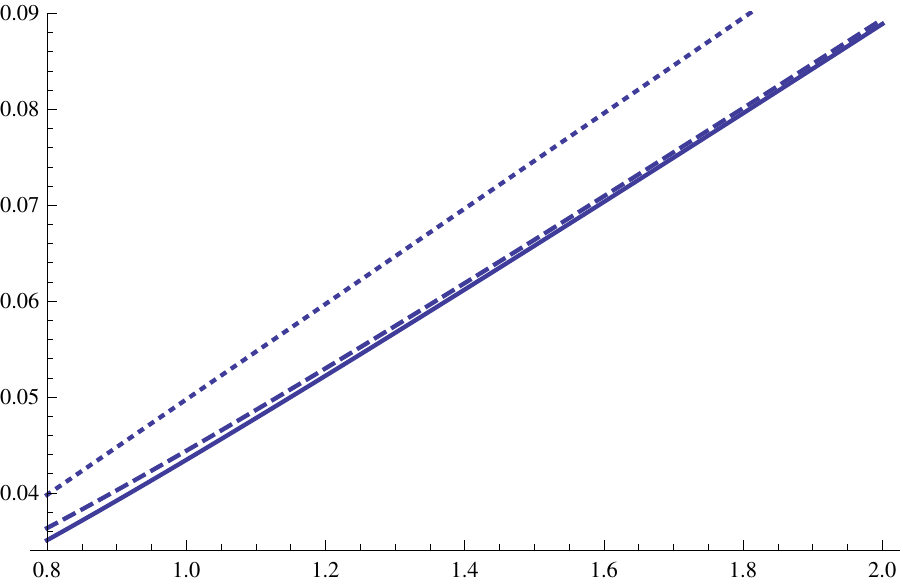}   
\caption{Local volatility squared $\protect\sigma_{\mathrm{loc}}^{2}(k,T)$
(solid curve) in the Heston model, with parameters $T=1$, $a=0.0428937$, $%
b=-0.6067$, $c=0.2928$, $s_0=1$, $v_0=0.0654$, $\protect\rho=-0.7571$. The
approximation~\eqref{LocVolLee} is dashed, and~\eqref{eq:heston asympt} is
dotted. }
\label{fig:heston}
\end{figure}
We will now show that the right hand side of~\eqref{LocVolLee} is indeed
asymptotically equivalent to the right hand side of~\eqref{eq:heston asympt}.
This requires us
to know that 
\begin{equation}
2\left. \frac{\partial }{\partial T}m( s,T) \right\vert _{s=\hat{s}( k,T)
}\sim 2k/\sigma,  \label{HestonClaim}
\end{equation}%
where $\sigma =\sigma ( T) $ is the so-called critical slope,
defined as  
\begin{align}
  \sigma (T) &=-\frac{\partial T^{\ast }}{\partial s}(s_{+}(T)), \label{eq:slope} \\
  T^{\ast }( s) &=\sup \{ t\geq 0:\mathbb{E}[ e^{sX_{t}}] <\infty \}. \notag
\end{align}
In fact, while the computation of the critical exponent $s_{+}$ in the 
Heston model requires simple numerics, the critical slope can be computed in
closed form \cite{FrGeGuSt11}; we have $\sigma ( T) =TR_{1}/R_{2}$, 
where~$R_i=R_i(b,c,\rho,s_+(T))$, $i=1,2$,
are defined in~\eqref{eq:R1}--\eqref{eq:R2}.

Since $\hat{s}( k,T) \rightarrow s_{+}( T) $ as $k\rightarrow \infty $,
the right hand side of~\eqref{LocVolLee} then satisfies%
\footnote{%
It is worth noting that $\sigma ( T) \sim const \times T$  as $T\rightarrow
\infty $. This suggests that $\sigma _{\mathrm{loc}}^{2}(  kT,T) $ admits a
non-degenerate limit as $T\rightarrow \infty $;  Gatheral's SVI limit of
Heston implied volatility was obtained in a similar  regime.} 
\begin{equation*}
\left. \frac{2\frac{\partial }{%
\partial T}m(s,T)}{s(s-1)}\right\vert _{s=\hat{s}(k,T)}
\sim \frac{2}{\sigma ( T) s_{+}( T) (
s_{+}( T) -1) }\times k, \qquad k\rightarrow \infty,
\end{equation*}
which is the formula from Theorem~\ref{thm:mainHeston}.
Let us now discuss validity of~\eqref{HestonClaim}. The argument which
follows nicely illustrates how formula~\eqref{LocVolLee} is used in
stochastic volatility models of affine type. First, $m( s,t) \approx
v_{0}\psi ( s,t) $ for a function $\psi $ for which we know\footnote{%
This follows from a straightforward analysis of the Riccati  equations~\cite%
{FrGeGuSt11}.} 
\begin{equation*}
\psi ( s,t) \sim \frac{1}{\frac{c^{2}}{2}( T^{\ast }( s) -t) }, \qquad
t\uparrow T^{\ast }( s),
\end{equation*}%
and also 
\begin{equation*}
\frac{\partial }{\partial t}\psi ( s,t) \sim \frac{1}{\frac{c^{2}}{2}(
T^{\ast }( s) -t) ^{2}}, \qquad t\uparrow T^{\ast }( s) .
\end{equation*}%
If we write $s_+=s_{+}( T) $ when $T$ is fixed, this translates to 
\begin{align}
m( s,t) &\sim \frac{v_{0}}{\frac{c^{2}}{2}\sigma ( s_+-s) }, \qquad
s\uparrow s_+,  \notag \\
\frac{\partial }{\partial s}m( s,t) &\sim \frac{v_{0}}{\frac{c^{2}}{2}\sigma
( s_+ -s) ^{2}}, \qquad s\uparrow s_+,  \label{eq:m} \\
\frac{\partial }{\partial T}m( s,T) &\sim \frac{v_{0}}{\frac{c^{2}}{2}(
\sigma (s_+ -s) ) ^{2}},\qquad s\uparrow s_+.  \label{eq:m_T}
\end{align}%
Equation (\ref{ApproximateSaddlePoint}) leads to $\hat{s}=s_+ -\beta
k^{-1/2} + o(k^{-1/2})$, since 
\begin{equation*}
\frac{\partial }{\partial s}m( s,t) \sim \frac{v_{0}}{\frac{c^{2}}{2}\sigma
( s_+ -\hat{s}) ^{2}}=k\implies s_+ -\hat{s}\sim\beta k^{-1/2}
\end{equation*}%
with $\beta =\frac{\sqrt{2v_{0}}}{c\sqrt{\sigma }}$. Substitution then yields%
\begin{equation*}
\frac{\partial }{\partial T}m( s,T) |_{s=\hat{s}}\sim \frac{v_{0}}{\frac{%
c^{2}}{2}\sigma ^{2}\beta ^{2}/k}=k/\sigma,
\end{equation*}%
which concludes our derivation of~\eqref{HestonClaim}.

\section{Some remarks on Dupire's formula for jump models}

\label{se:jumps}

As discussed in the introduction, a direct application of Dupire's
formula is not easy to justify in the presence of jumps. Even so,
given the industry practice of applying Dupire's formula to any given
call price surface, we now discuss
what happens when applying formula~\eqref{LocVolLee} to some examples of jump models.
\begin{example}[Double exponential L\'{e}vy]
\label{ex:kou}  For zero drift, the mgf is given by~\cite{CoTa04} 
%TODO maybe mention page, and typos
\begin{equation*}
M(s,T) = \exp\left(T\left(\frac{\sigma^2 s^2}{2} + \lambda\left(  \frac{%
\lambda_+ p}{\lambda_+ - s} +  \frac{\lambda_-(1-p)}{\lambda_- + s}%
\right)\right)\right).  
\end{equation*}
The critical moment is $s_+=\lambda_+$, and the saddle point is located at  
\begin{equation}  \label{eq:kou sp}
\hat{s} \approx s_+ - \sqrt{\frac{\lambda \lambda_+ p T}{k}}.
\end{equation}
Formula~\eqref{LocVolLee} thus yields  
\begin{equation}  \label{eq:kou asympt}
\sigma_{\mathrm{loc}}^{2}(k,T) \approx \frac{2\sqrt{\lambda p}}{\sqrt{\lambda_+ T}%
(\lambda_+ -1)} k^{1/2}.
\end{equation}
In Figure~\ref{fig:kou}, the fit of~\eqref{eq:kou asympt}  is not
satisfactory (the dotted curve).  Similarly to the Heston model, the
approximation~\eqref{eq:kou asympt}  has on~$O(1)$ error term, whereas~%
\eqref{LocVolLee} seems to have~$o(1)$,  and gives a very good estimate.  
%TODO ist die andere Approx wirklich so schlecht? Formel nachrechnen
\begin{figure}[ht]
\centering
\includegraphics[width=8truecm]{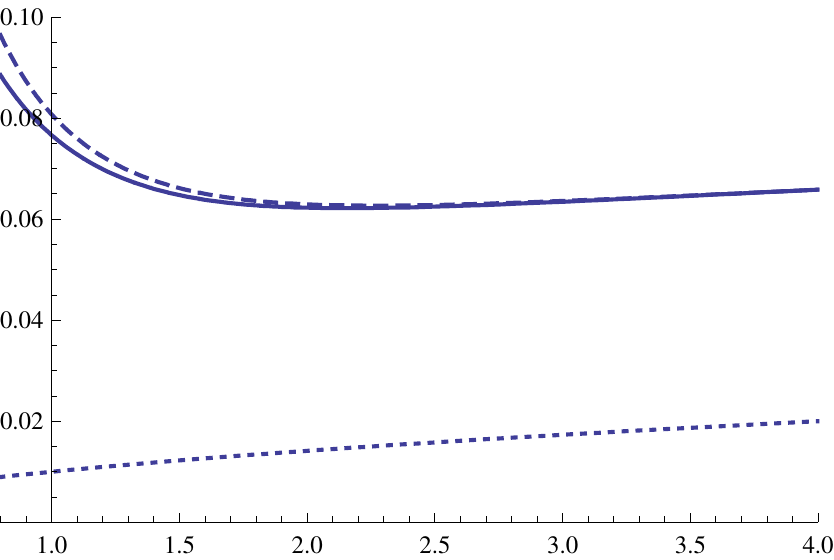}   
\caption{Local volatility squared $\protect\sigma_{\mathrm{loc}}^{2}(k,T)$
(solid curve) in the double exponential L\'{e}vy model, with parameters $T=1$%
, $\protect\sigma=0.2$, $\protect\lambda=10$, $p=0.3$, $\protect\lambda_- =
25$, $\protect\lambda_+ = 50$. The approximation~\eqref{LocVolLee} is
dashed, whereas~\eqref{eq:kou asympt} is dotted. }
\label{fig:kou}
\end{figure}

As mentioned in Section~\ref{se:intro}, one expects local volatility to explode for $%
T\downarrow 0$: We have $\sigma _{\mathrm{loc}}^{2}(k,T) \approx  const \times
T^{-1/2}$ in the double exponential L\'{e}vy model.  To see this, note that
in L\'evy models the saddle point  $\hat{s}(k,T)$ is a function of $k/T$,
and that the saddle point method  works for $T\downarrow 0$ as well as for $%
k\to\infty$. Therefore,~\eqref{eq:kou asympt}  is true for fixed~$k$ and $%
T\downarrow 0$, too.
\end{example}

\begin{example}[Variance gamma]
\label{ex:var gamma}  The mgf is given by~\cite{MaCaCh98}  
\begin{equation*}
M(s,T) = \left(\frac{1}{1-\theta\nu s-\tfrac12 \sigma^2 \nu s^2}
\right)^{T/\nu}.
\end{equation*}
We assume that $T>\nu/2$, which guarantees that the log-price has a density,
and hence that call prices are~$C^2$ (see Example~1 in~\cite{CoVo05}).
The critical moment is
\begin{equation*}
s_+ = \frac{\sqrt{2\nu\sigma^2+\nu^2 \theta^2}-\nu \theta}{\nu \sigma^2}.  
\end{equation*}
Since we have  
\begin{equation*}
m(s,T) = \log M(s,T) \sim \frac{T}{\nu} \log \frac{1}{s_+-s},  
\end{equation*}
the saddle point satisfies  
\begin{equation*}
\hat{s} \approx s_+ - \frac{T}{\nu k}.  
\end{equation*}
By~\eqref{LocVolLee}, we thus have  
\begin{equation}  \label{eq:varg asympt}
\sigma_{\mathrm{loc}}^{2}(k,T) \approx \frac{2\log (k/T)}{\nu s_+(s_+-1)}.
\end{equation}
According to Figure~\ref{fig:var gamma}, this  approximation kicks in only
for fairly large values of~$k$.  As in Section~\ref{se:heston} and Example~%
\ref{ex:kou},  an improved estimate is obtained by  using~\eqref{LocVolLee}
directly.   
\begin{figure}[ht]
\centering
\includegraphics[width=8truecm]{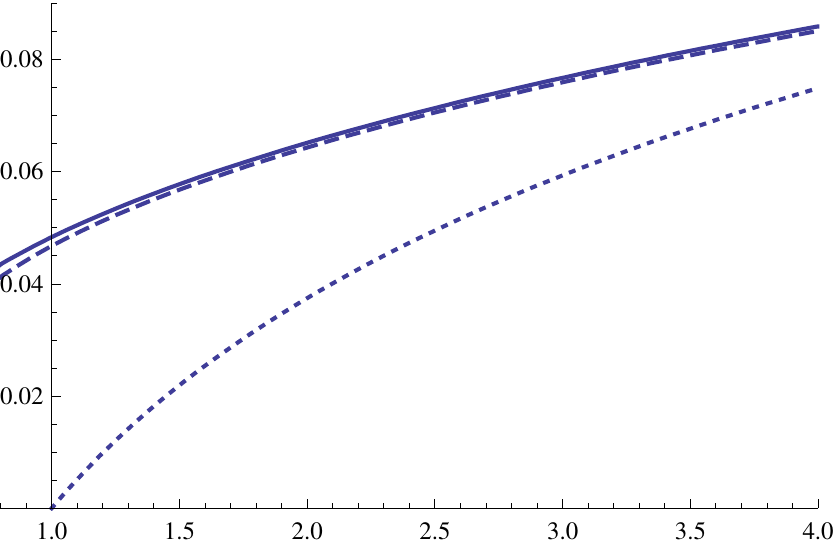}   
\caption{Local volatility squared $\protect\sigma_{\mathrm{loc}}^{2}(k,T)$
(solid curve) in the variance gamma model, with parameters $T=1$, $\protect%
\sigma=0.261652$, $\protect\theta=-0.218033$, $\protect\nu = 0.0552584$. The
approximation~\eqref{LocVolLee} is dashed, and~\eqref{eq:varg asympt} is
dotted. }
\label{fig:var gamma}
\end{figure}
%The numerical parameters were chosen such that formula~\eqref{eq:dupire}  is
%well-defined (Cf.\ the comments at the end of Section~\ref{se:intro}.  The
%smoothness of $C(\cdot,T)$ increases w.r.t.~$T$;  see~\cite[Example~1]%
%{CoVo05} for a discussion.)
%
%             Useless, since local vol undefined:
%For small~$T$, the squared local vol behaves like $\log(1/T)$. This is
%derived similarly  as in Example~\ref{ex:kou}.
\end{example}

\begin{example}[Normal inverse Gaussian]
\label{ex:NIG}
This is an example where condition~\eqref{eq:blowup} is violated, and
our formula~\eqref{LocVolLee} does not hold.
  The mgf  
\begin{equation*}
M(s,T) = \exp\left( \delta T\left( \sqrt{\alpha^2-\beta^2}  -\sqrt{%
\alpha^2-(\beta+s)^2} \right) \right)  
\end{equation*}
has no blow-up at the critical moment  
\begin{equation*}
s_+ = \alpha - \beta,  
\end{equation*}
but a square-root type singularity, with local expansion  
\begin{equation}  \label{eq:M NIG}
M(s,T) \approx e^{\delta T\sqrt{\alpha^2-\beta^2}}\left(1 - \delta T\sqrt{%
2\alpha} \sqrt{s_+-s} \right).
\end{equation}
It is still true that $\sigma_{\mathrm{loc}}^2(k,T)$ asymptotically depends,
via~\eqref{eq:dup frac},  on the local behavior of~$M(s,T)$ near~$s_+$.
However, the approximation~\eqref{LocVolLee} hinges  on the \emph{first}
term of the local expansion of $M(s,T)$. It therefore fails to capture  the
asymptotics of $\sigma_{\mathrm{loc}}^2(k,T)$, which  depend on the first 
\emph{singular} term (the term $\sqrt{s_+-s}$ in~\eqref{eq:M NIG}).  An
analysis can be done by a Hankel contour approach~\cite{Ford60}, and yields  
\begin{equation*}
\sigma_{\mathrm{loc}}^2(k,T) \approx \frac{2\left(1 + \delta T\sqrt{%
\alpha^2-\beta^2} \right)} {s_+(s_+ -1)}.  
\end{equation*}
The numerical fit is not very good, though, and further terms should be
computed  for improved accuracy.  We propose to return to this model and the
more general GH (generalized hyperbolic)  model in a future study.

The fact that $\sigma_{\mathrm{loc}}^2(k,T)$ converges to a constant might be
understood by comparing the  NIG marginals with those of Heston's in the
time $T\rightarrow \infty $  regime (this link is made precise in \cite{KR11}%
). In particular, the result  is then consistent with the Heston asymptotics~%
\eqref{eq:heston asympt}  of local vol; note that the right hand side
of~\eqref{eq:heston asympt} is $O(1/T)$ for $T\to\infty$.
\end{example}

\section{Conclusions}

We propose a new formula that expresses local volatility for extreme strikes as
a computable function of commonly available model information. In the Heston
model this leads to a proof that local volatility (squared) behaves
asymptotically linear in log-strike
(which is qualitatively similar to Lee's result~\cite{Le04a} for implied volatility).

Although we suspect that this Lee-type behavior remains true for models similar
enough to Heston (e.g.\ local stochastic volatility models with a Heston backbone~\cite{HL09}),
qualitatively different behavior is seen in models with jumps. We derived this
by applying our generic approximation formula~\eqref{LocVolLee}, supported by
numerical examples.

While this enhances our knowledge of local volatility in a variety of models, it
also has a clear impact on calibration of local volatility to market
data. Indeed,
liquid option data is typically available only in a restricted range
of strikes and
maturities; our results can then be used to extrapolate local
volatility in a way
that is consistent with Heston stochastic volatility or other chosen models.
In particular, this approach avoids any arbitrage possibilities
introduced by ad-
hoc specifications of the implied volatility surface. We also believe the present
methodology will turn out useful in the calibration of local
stochastic volatility
models to market smiles.
\bigskip

\textbf{Acknowledgment.} We thank Rama Cont and Jim Gatheral for helpful
discussions.

\begin{appendix}
\section{Proof of Theorem~\ref{thm:mainHeston} (local vol approximation
for the Heston model)}
\label{se:app heston}

By the exponential decay of the Heston mgf towards $\pm i \infty$,
the formulas \eqref{eq:call}--\eqref{eq:dup frac} are correct for the Heston
model. For the saddle point analysis of~\eqref{eq:dup frac}, we employ
the approximate saddle point
\[
  \hat{s}_{\mathrm{approx}}(k) := s_+ - \beta k^{-1/2},
\]
obtained by using~\eqref{eq:m} in~\eqref{ApproximateSaddlePoint}.
(Recall that $\beta=\sqrt{2v_0}/c\sqrt{\sigma}$, and that~$\sigma$ denotes
the critical slope defined in~\eqref{eq:slope}.)
This approximate saddle may be used for both integrals in~\eqref{eq:dup frac}.
As for the denominator, this was carried out in detail in~\cite{FrGeGuSt11},
where an expansion of the Heston density was determined.
The analysis of the numerator in~\eqref{eq:dup frac} is similar, except
that a new tail estimate is required.
But first we discuss the local expansion around the saddle point. Let us fix a number
$\alpha\in(\tfrac23,\tfrac34)$ and define $h(k)=k^{-\alpha}$.
Then, in the central range $|s-\hat{s}_{\mathrm{approx}}(k)| \leq h(k)$,
we have
\begin{align*}
  \frac{1}{s(s-1)} &= \frac{1}{s_+(s_+ -1)} + O(s_+ -s) \\
  &= \frac{1}{s_+(s_+ -1)} \left(1 + O(k^{-1/2}) \right)
\end{align*}
and (cf.~\eqref{eq:m_T})
\begin{align*}
  2\frac{\partial}{\partial T} m(s,T) &= \frac{2\beta^2}{\sigma(s_+-s)^2}
    + O\left( \frac{1}{s_+ -s} \right) \\
  &= \frac{2\beta^2}{\sigma} (\beta k^{-1/2} + O(k^{-\alpha}) )^{-2}
    + O(k^{-1/2}) \\
  &= \frac{2k}{\sigma} (1 + O(k^{1/2-\alpha})).
\end{align*}
Therefore, the local expansions of the two integrands in~\eqref{eq:dup frac}
agree, up to a factor that is given by
\begin{equation}\label{eq:new fac}
  \frac{2\partial_T m(s,T)}{s(s-1)} = \frac{2k}{\sigma s_+(s_+-1)} (1 + O(k^{1/2-\alpha})),
\end{equation}
where the error term holds uniformly w.r.t.\ the integration variable~$s$.
According to Theorem~1.2 of~\cite{FrGeGuSt11}, we have
\begin{equation}\label{eq:density as}
  \frac{1}{2i\pi} \int_{\hat{s}_{\mathrm{approx}}-ih(k)}%
  ^{\hat{s}_{\mathrm{approx}}+ih(k)}
    e^{-ks} M(s,T) ds
  \sim A_1 e^{(1-A_3)k + A_2 \sqrt{k}} k^{-3/4+a/c^2}
\end{equation}
for certain constants $A_1$, $A_2=2\beta$, and $A_3=s_+ +1$.
Analogously, we derive from~\eqref{eq:new fac} that
\begin{multline}\label{eq:tail new fac}
  \frac{1}{2i\pi}  \int_{\hat{s}_{\mathrm{approx}}-ih(k)}%
  ^{\hat{s}_{\mathrm{approx}}+ih(k)}
     \frac{2\partial_T m(s,T)}{s(s-1)} e^{-ks} M(s,T) ds \\
    \sim \frac{2k}{\sigma s_+(s_+-1)} \times
    A_1 e^{(1-A_3)k + A_2 \sqrt{k}} k^{-3/4+a/c^2}.
\end{multline}
Dividing~\eqref{eq:tail new fac} by~\eqref{eq:density as} shows our
claim~\eqref{eq:heston asympt}, provided that the tails
$|s-\hat{s}_{\mathrm{approx}}(k)| > h(k)$ of the integrals can be discarded.
For the denominator of~\eqref{eq:dup frac}, this was shown in
Lemma~A.3 of~\cite{FrGeGuSt11}.
So we proceed with the numerator. We consider only the upper tail, as the
lower one is handled by symmetry.
By Lemma~A.3 of~\cite{FrGeGuSt11}, there is a constant $B>0$ such that
\begin{equation}\label{eq:inner tail}
  \left| \int_{\hat{s}_{\mathrm{approx}}+ i h(k)}^{\hat{s}_{\mathrm{approx}}+iB}
     e^{-ks} M(s,T)  ds  \right| \leq
    e^{(1-A_3)k} \exp(A_2 \sqrt{k} - \tfrac12 \beta^{-1}
    k^{3/2-2\alpha} + O(\log k)).
\end{equation}
From~\eqref{eq:m} we obtain
\[
  \left|  \frac{\partial_T m(s,T)}{s(s-1)} \right| \leq const \times k
\]
for all~$s$ on the contour
in~\eqref{eq:inner tail}. This estimate can be absorbed into
the factor $\exp(O(\log k))$ in~\eqref{eq:inner tail},
so that we conclude
\begin{multline}\label{eq:inner tail 2}
  \left| \int_{\hat{s}_{\mathrm{approx}}+ I h(k)}^{\hat{s}_{\mathrm{approx}}+iB}
      \frac{\partial_T m(s,T)}{s(s-1)} e^{-ks} M(s,T)  ds \right| \\
    \leq  e^{(1-A_3)k} \exp(A_2 \sqrt{k} - \tfrac12 \beta^{-1}
    k^{3/2-2\alpha} + O(\log k)).
\end{multline}
This grows slower than~\eqref{eq:tail new fac} (compare the relevant factors
$k^{-3/4+a/c^2}$ resp.\ $\exp(- \tfrac12 \beta^{-1} k^{3/2-2\alpha})$).
As for $\Im(s) > B$, it was shown in~\cite{FrGeGuSt11} (Lemma~A.2)
that
\[
  \left| \int_{\hat{s}_{\mathrm{approx}}+ i B}^{\hat{s}_{\mathrm{approx}}+i\infty}
    e^{-ks} M(s,T)  ds \right|
    =O(\exp((1-A_3)k + \beta \sqrt{k})).
\] 
This was deduced from the exponential decay of~$M(s,T)$
for large $\Im(s)$ (Lemma A.1 in~\cite{FrGeGuSt11}). The following lemma implies
that the new factor $\partial_T m(s,T)/(s(s-1))$
grows only polynomially, so that the exponential decay of the integrand persists
for the numerator of~\eqref{eq:dup frac}.
This finishes the proof of Theorem~\ref{thm:mainHeston}.

To state the lemma, recall that $m(s,t) = \phi(s,t) + v_0 \psi(s,t)$, where $\phi$ and $\psi$
satisfy the Riccati equations
\begin{align*}
  \dot{\phi} &= a \psi, \qquad \phi(0) = 0, \\
  \dot{\psi} &= \tfrac12 (s^2-s)+\tfrac12 c^2 \psi^2
  +b\psi + s\rho c \psi, \qquad \psi(0) = 0.
\end{align*}
We have to show that $\dot{m}$ grows only polynomially
as $\Im(s) \to \infty$.
Because of the Riccati equations, it suffices to show this for~$\psi$.
Let us write $\psi = f + i g$ and $s = \xi + i y$.
\begin{lemma}
Let $T>0$, and assume that the real part~$\xi$ of~$s$ stays
bounded in some interval $1\leq \xi \leq \xi_{\mathrm{max}}$.
Then, there are positive constants $C_{i,T}$ ($i=1,2,3,4\,)$ such
that \thinspace for $y\geq y_{0}$, where $y_{0}$ depends only on $\xi _{\max
}$ and the other (fixed)\ model parameters of the Heston model,%
\begin{align*}
-C_{3,T}y^{2} &\leq f(t)\leq -C_{1,T}y, \\
-C_{4,T}y^{3} &\leq g(t)\leq C_{2,T}\,y.
\end{align*}%
In fact, we can take%
\begin{align*}
C_{1,T} &= 1/\left( 3c\right),  \\
C_{2,T} &= \frac{1}{2}\left( 2\xi_{\mathrm{max}} -1\right) T, \\
C_{3,T} &= T\left( 1+\frac{c^{2}}{2}C_{2,T}^{2}\right),  \\
C_{4,T} &= 2C_{3,T}Tc^{2}C_{2,T}.
\end{align*}
\end{lemma}
\begin{proof}
It follows from the proof of Lemma~A.1 in~\cite{FrGeGuSt11}
that
(e.g.\ with $C_{1,T}:=T\theta =\frac{1}{c}\sqrt{1/6}%
\leq \frac{1}{3c}$)%
\begin{equation*}
f(t) \leq -T\theta y=-\frac{1}{c}\sqrt{1/6}y\leq -\frac{1}{3c}%
y =: -C_{1,T}y.
\end{equation*}%
We next provide a similar upper estimate for~$g$. To this end we first show that 
$g=g( t) $ remains $\geq 0$ for all times $t>0$. The differential
equation for $g$,%
\begin{equation*}
\dot{g}=\frac{1}{2}\left( 2\xi y-y\right) +c^{2}fg-\gamma g, \qquad
g(0) =0,
\end{equation*}%
implies the first order Euler estimate
\begin{align*}
g\left( t\right)  &=g(0) +\left\{ \frac{1}{2}\left( 2\xi
y-y\right) +c^{2}f(0) g(0) -\gamma g(0)
\right\} t+o(t)  \\
&\underset{>0}{=\underbrace{\frac{1}{2}\left( 2\xi y-y\right) }t+o(
t) },
\end{align*}%
and hence~$g$ is positive (even strictly so) on some interval $\left(
0,\varepsilon _{1}\right) $. Assume this interval is maximal in the sense
that $g( \varepsilon _{1}) =0$ and $g$ is (strictly) negative on
some further interval $\left( \varepsilon _{1},\varepsilon _{2}\right) $.
Clearly then $\dot{g}( \varepsilon _{1}) \leq 0$, which
contradicts the information from the differential equation: indeed, using $%
g( \varepsilon _{1}) =0$, we obtain the contradiction%
\begin{equation*}
\dot{g}( \varepsilon _{1}) =\underset{>0}{\underbrace{\frac{1}{2}%
\left( 2\xi y-y\right) }}.
\end{equation*}%
The observation that $g\geq 0$ is useful to us, since it leads, together with 
$f\leq -C_{1,T}y$ and $\gamma \geq 0$, to the differential inequality%
\begin{align*}
\dot{g} &=\frac{1}{2}\left( 2\xi y-y\right) +c^{2}fg-\gamma g \\
&\leq \frac{1}{2}\left( 2\xi y-y\right) -\left( c^{2}C_{1,T}+\gamma \right)
g \\
&\leq \frac{1}{2}\left( 2\xi y-y\right),
\end{align*}%
and hence to the upper estimate%
\begin{equation*}
\forall 0\leq t\leq T:g( t) \leq \frac{1}{2}
\left( 2\xi_{\mathrm{max}} -1\right)
T\times y =: C_{2,T}y.
\end{equation*}%
We can feed this upper estimate on $g$ back in the differential equation
for~$f$ to obtain a lower estimate%
\begin{align*}
\dot{f} &=\frac{1}{2}\left( \xi ^{2}-y^{2}-\xi \right) +\frac{c^{2}}{2}%
\left( f^{2}-g^{2}\right) -\gamma f \\
&\geq \frac{1}{2}\left( \xi ^{2}-y^{2}-\xi \right) +\frac{c^{2}}{2}f^{2}-%
\frac{c^{2}}{2}C_{2,T}^{2}y^{2}-\gamma f \\
&=-\frac{1}{2}\left( 1+c^{2}C_{2,T}^{2}\right) y^{2}+\frac{1}{2}\left( \xi
^{2}-\xi \right) -\gamma f+\frac{c^{2}}{2}f^2 \\
&\geq -\frac{1}{2}\left( 1+c^{2}C_{2,T}^{2}\right) y^{2}+\frac{1}{2}\left(
\xi ^{2}-\xi \right) -\gamma f \\
&\geq -\left( 1+\frac{c^{2}}{2}C_{2,T}^{2}\right) y^{2}-\gamma f,
\end{align*}%
where in the last step we assume that $y$ is large enough so that the extra
amount subtracted (at least: $\frac{1}{2}y^{2}$) is larger than $\frac{1}{2}%
\left( \xi ^{2}-\xi \right) $, which remains bounded.
%(After all $\xi $ is confined such that $\xi +iy$ remains in the strip of analyticity.)
We also
know that $f( t) \leq -C_{1,T}y\leq 0$ for all $0\leq t\leq T$.
It follows that $-\gamma f\geq 0$ and omission leads to our final lower
bound on $\dot{f}$, namely%
\begin{equation*}
\dot{f}\geq -\left( 1+\frac{c^{2}}{2}C_{2,T}^{2}\right) y^{2}.
\end{equation*}%
This entails immediately%
\begin{equation*}
f( t) \geq -T\left( 1+\frac{c^{2}}{2}C_{2,T}^{2}\right)
y^{2} =: -C_{3,T}y^{2}.
\end{equation*}%
At last, we need a lower bound on $g$. Again, we look for a suitable
differential inequality. Since $g\geq 0$,%
\begin{align*}
\dot{g} &=\frac{1}{2}\left( 2\xi y-y\right) +c^{2}fg-\gamma g \\
&\geq \frac{1}{2}\left( 2\xi y-y\right) -c^{2}\left\vert f\right\vert
g-\gamma g \\
&\geq \frac{1}{2}\left( 2\xi -1\right) y-\left( C_{3,T}y^{2}c^{2}+\gamma
\right) g \\
&\geq \frac{1}{2}\left( 2\xi -1\right) y-2C_{3,T}y^{2}c^{2}g,
\end{align*}%
for $y$ large enough such that the additional subtraction of $%
C_{3,T}y^{2}c^{2}$ takes care of~$\gamma $. Using the upper estimate on 
$g$ (linear in $y$), and the fact that $\left( 2\xi -1\right) \geq 0$,
we conclude
\begin{equation*}
\dot{g}\geq -2C_{3,T}y^{2}c^{2}C_{2,T}y.
\end{equation*}%
It immediately follows (since $g( 0) =0$) that 
\begin{equation*}
\forall 0\leq t\leq T:g( t) \geq
-2C_{3,T}Tc^{2}C_{2,T}y^{3} =: -C_{4,T}y^{3}.
\end{equation*}
\end{proof}

\end{appendix}

\bibliographystyle{siam}
\bibliography{../gerhold}

\begin{thebibliography}{10}

\bibitem{AnPi07}
{\sc L.~B.~G. Andersen and V.~V. Piterbarg}, {\em Moment explosions in
  stochastic volatility models}, Finance Stoch., 11 (2007), pp.~29--50.

\bibitem{AvBOBuFr03}
{\sc M.~Avellaneda, D.~Boyer-Olson, J.~Busca, and P.~K. Friz}, {\em
  Reconstruction of volatility: Pricing index options using the
  steepest-descent approximation}, Risk,  (October 2002), pp.~91--95.

\bibitem{BeCo10}
{\sc A.~Bentata and R.~Cont}, {\em Forward equations for option prices in
  semimartingale models}.
\newblock Preprint, available at http://arxiv.org/abs/1001.1380, 2010.

\bibitem{BeBuFl04}
{\sc H.~Berestycki, J.~Busca, and I.~Florent}, {\em Computing the implied
  volatility in stochastic volatility models}, Comm. Pure Appl. Math., 57
  (2004), pp.~1352--1373.

\bibitem{BiGoTe87}
{\sc N.~H. Bingham, C.~M. Goldie, and J.~L. Teugels}, {\em Regular variation},
  vol.~27 of Encyclopedia of Mathematics and its Applications, Cambridge
  University Press, Cambridge, 1987.

\bibitem{CaGeMaYo04}
{\sc P.~Carr, H.~Geman, D.~P. Madan, and M.~Yor}, {\em From local volatility to
  local {L}\'evy models}, Quant. Finance, 4 (2004), pp.~581--588.

\bibitem{CaMa99}
{\sc P.~Carr and D.~P. Madan}, {\em Option valuation using the fast {F}ourier
  transform}, Journal of Computational Finance, 3 (1999), pp.~463--520.

\bibitem{CoTa04}
{\sc R.~Cont and P.~Tankov}, {\em Financial modelling with jump processes},
  Chapman \& Hall/CRC Financial Mathematics Series, Chapman \& Hall/CRC, Boca
  Raton, FL, 2004.

\bibitem{CoVo05}
{\sc R.~Cont and E.~Voltchkova}, {\em Integro-differential equations for option
  prices in exponential {L}\'evy models}, Finance Stoch., 9 (2005),
  pp.~299--325.

\bibitem{deBr58}
{\sc N.~G. de~Bruijn}, {\em Asymptotic methods in analysis}, Bibliotheca
  Mathematica. Vol. 4, North-Holland Publishing Co., Amsterdam, 1958.

\bibitem{Du94}
{\sc B.~Dupire}, {\em Pricing with a smile}, Risk, 7 (1994), pp.~18--20.

\bibitem{Du96}
\leavevmode\vrule height 2pt depth -1.6pt width 23pt, {\em A unified theory of
  volatility}.
\newblock Working paper, {P}aribas, 1996.

\bibitem{Ford60}
{\sc W.~B. Ford}, {\em Studies on divergent series and summability and the
  asymptotic developments of functions defined by {Maclaurin} series}, Chelsea
  Publishing Company, 3rd~ed., 1960.
\newblock (From two books originally published in 1916 and 1936.).

\bibitem{FrGeGuSt11}
{\sc P.~Friz, S.~Gerhold, A.~Gulisashvili, and S.~Sturm}, {\em On refined
  volatility smile expansion in the {H}eston model}.
\newblock To appear in Quantitative Finance, 2011.

\bibitem{Ga06}
{\sc J.~Gatheral}, {\em {The Volatility Surface, A Practitioner's Guide}},
  Wiley, 2006.

\bibitem{HL05}
{\sc P.~Henry-Labordere}, {\em A general asymptotic implied volatility for
  stochastic volatility models}.
\newblock Available at SSRN: http://ssrn.com/abstract=698601, 2005.

\bibitem{HL09}
{\sc P.~Henry-Labordere}, {\em Calibration of local stochastic volatility
  models to market smiles: A {M}onte-{C}arlo approach}, Risk,  (September
  2009).
\newblock Available at SSRN: http://ssrn.com/abstract=1493306.

\bibitem{He93}
{\sc S.~Heston}, {\em A closed-form solution for options with stochastic
  volatility with applications to bond and currency options}, Review of
  Financial Studies, 6 (1993), pp.~327--343.

\bibitem{KR11}
{\sc M.~Keller-Ressel}, {\em Moment explosions and long-term behavior of affine
  stochastic volatility models}, Math. Finance, 21 (2011), pp.~73--98.

\bibitem{Le04a}
{\sc R.~W. Lee}, {\em The moment formula for implied volatility at extreme
  strikes}, Math. Finance, 14 (2004), pp.~469--480.

\bibitem{MaCaCh98}
{\sc D.~Madan, P.~Carr, and E.~Chang}, {\em The variance gamma process and
  option pricing}, European Finance Review, 2 (1998), pp.~79--105.

\bibitem{Pi06}
{\sc V.~Piterbarg}, {\em Markovian projection method for volatility
  calibration}.
\newblock Available at SSRN: http://ssrn.com/abstract=906473, 2006.

\end{thebibliography}

\end{document}